\documentclass[%
 reprint,
nofootinbib,
 amsmath,amssymb,
 aps,
]{revtex4-2}

\let\chap\S

\usepackage{graphicx}
\usepackage{dcolumn}
\usepackage{bm}
\usepackage{physics}
\usepackage{amsthm}

\DeclareMathOperator*{\argmax}{arg\,max}

\newcommand{\X}{\mathcal{X}}
\newcommand{\Y}{\mathcal{Y}}
\newcommand{\Z}{\mathcal{Z}}
\newcommand{\Rho}{\mathcal{R}}
\renewcommand{\L}{\mathcal{L}}
\renewcommand{\H}{\mathcal{H}}

\newcommand{\N}{\mathcal{N}}
\newcommand{\F}{\mathcal{F}}
\renewcommand{\S}{\mathcal{S}}
\renewcommand{\P}{\mathcal{P}}

\newtheorem{theorem}{\bf Theorem}
\newtheorem{proposition}{\bf Proposition}

\newtheorem{corollary}{\bf Corollary}

\newtheorem{definition}{\bf Definition}

\begin{document}

\preprint{APS/123-QED}

\title{Optimal Universal Quantum Encoding for Statistical Inference}

\author{Farhad Farokhi}
 \email{farhad.farokhi@unimelb.edu.au}
\affiliation{Department of Electrical and Electronic Engineering, The University of Melbourne, Parkville, VIC 3010, Australia
}%

\date{\today}

\begin{abstract}
Optimal encoding of classical data for statistical inference using quantum computing is investigated. A universal encoder is sought that is optimal for a wide array of statistical inference tasks. Accuracy of any statistical inference is shown to be upper bounded by a term that is proportional to maximal quantum leakage from the classical data, i.e., the input to the inference model, through its quantum encoding. This demonstrates that the maximal quantum leakage is a universal measure of the quality of the encoding strategy for statistical inference as it only depends on the quantum encoding of the data and not the inference task itself. The optimal universal encoding strategy, i.e., the encoding strategy that maximizes the maximal quantum leakage, is proved to be attained by pure states. When there are enough qubits, basis encoding is proved to be universally optimal. An iterative method for numerically computing the optimal universal encoding strategy is presented. 
\end{abstract}

\maketitle


\section{Introduction}
Encoding classical data, e.g., traditional datasets, into quantum systems often forms the first step in quantum computing and communication. In quantum machine learning~\cite{PhysRevA103032430}, data needs to be encoded into quantum states so that it can be fed into parameterized quantum circuits for predictions. In quantum key distribution~\cite{bennett2014quantum}, random bits must be encoded in non-orthogonal quantum states for communication. In quantum communication~\cite{haselgrove2005optimal}, messages are encoded by appropriate code-words to ensure high fidelity over noisy channels.

Not all encoding strategies are equal. There are many methods for encoding classical data into quantum systems, such as basis encoding, angle encoding, quantum random access memory encoding, and amplitude encoding~\cite{10555535110653511068, weigold2021encoding}. There have been many attempts at systematically comparing encoding strategies and establishing optimal ones. The effect of commonly used data-encoding mechanisms on the expressivity of quantum machine learning models was explored in~\cite{PhysRevA103032430}. Optimal encoding to achieve maximal fidelity in communication over spin systems was derived in~\cite{haselgrove2005optimal}. Optimal encoding and retrieval of classical data in quantum domain was developed in~\cite{elron2007optimal}. Quantum information theory, especially relative entropy, was used to analyze quantum encoding over a  family of communication channels~\cite{korzekwa2022encoding}. Incompatibility of states was used to understand security of various encoding strategies in quantum key sharing mechanisms~\cite{mitra2021optimal}. This problem was approached again using information leakage with gentle measurements in~\cite{farokhi2024measuring}. However, there has been no information-theoretic approach to analyze quantum encoding strategies in statistical inference and to develop universal encoding strategies. Universality here refers to encoders that are optimal for a wide array of statistical inference tasks as opposed to problem-specific encoders. This is the topic of the current letter.

In this letter, we focus on information-theoretical analysis of statistical inference using quantum computing. We consider an inference model to guess or estimate the output $Z$ based on access to input $X$. Random variables $X$ and $Z$ are assumed to be jointly distributed. For instance, in image recognition, the input $X$ is an image or its pixelated representation while the output $Z$ is the category to which the image belongs (e.g., cat or dog) or the content of the image (e.g., location of obstacles). Note that the inference model does not have access to the output $Z$ and must guess the output based on the input $X$ and historical data (often called training dataset).
In this letter, we focus on the large-data regime, where we can assume that the probability distribution of the data is known, i.e., it is accurately estimated based on many available data points. 
The outcome of the inference model is denoted by $\widehat{Z}$. In statistical inference, the quality of an inference model is synonymous with its accuracy $\mathbb{P}\{\widehat{Z}=Z\}$. For statistical inference using a quantum-computing procedure, as the first step, the classical data $X$ must be encoded into a quantum state. The quantum state is then processed by arbitrary quantum channels, which can model unitary gates used in quantum computing~\cite{nielsen2010quantum} and more general noisy transformations~\cite{wilde2013quantum}. Subsequently, measurements are taken and post-processed to generate the output estimate $\widehat{Z}$. We prove that, irrespective of the objective of the inference task, the accuracy of the quantum computing procedure is constrained by a term that is proportional with the maximal quantum leakage~\cite{Farokhi_PRA} from the classical data $X$ through its quantum encoding. This demonstrates that the maximal quantum leakage is a universal measure of the quality of the encoding strategy. Interestingly, maximal quantum inference is also independent of the distribution of the input $X$ (and is only a function of the quantum encoding of the classical data). Therefore, the optimal universal encoder, i.e., the encoder that maximizes the maximal quantum leakage, is independent of the inference problem at hand. This feature signifies the universality of the optimal universal encoder.
These observations show that the number of qubits required for statistical inference must be larger than $\log_2(|\X|)/2$, where $\X$ is the support set of discrete input random variable $X$, to not artificially constrain the performance of the inference model. We prove that the optimal universal encoding strategy is composed of pure states. This is an important revelation as often quantum computing procedures rely on pure states. Furthermore, when there are enough qubits, i.e., the number of qubits is higher than $\log_2(|\X|)$, the basis encoding is proved to be the   optimal universal encoder. Finally, an iterative method for numerically computing the optimal universal encoding strategy in all other cases is presented. This procedure relies on subgradient ascent for maximization of the maximal quantum leakage.

\section{Statistical Inference via Quantum Computing}
Consider a statistical inference problem with jointly distributed discrete random variables $X\in\X$, referred to as the input in what follows, and $Z\in\Z$, referred to as the output. The aim is to develop an inference model to guess or estimate the output $Z$ based on only access to $X$ and the distribution of the underlying random variables. The outcome of the inference model is denoted by $\widehat{Z}\in\Z$. The accuracy of the statistical inference model is given by the probability of the event that $\widehat{Z}$ and $Z$ coincide $\mathbb{P}\{\widehat{Z}=Z\}$. In this letter, we use a quantum computing framework for  the statistical inference. This framework is composed of four parts: quantum encoding, processing, measurement, and classical post-processing. In what follows, we explain each part briefly.

\textit{Encoding}: Here, we convert or encode the raw classical data, i.e., the realization of random variable $X$, into a quantum system. For the purpose of this letter, a quantum system is modelled by a finite-dimensional Hilbert space $\H$. The set of positive semi-definite linear operators on $\H$ is denoted by $\P(\H)$ while the set of density operators, i.e., positive semi-definite linear operators on $\H$ with unit trace, is denoted by $\S(\H)$. We use density operators to model the state of a quantum system. The quantum encoding entails preparing quantum system $A$ in state $\rho^x\in\S(\H)$ if the input is $X=x\in\X$ is realized. The ensemble $\Rho:=\{\rho^x\}_{x\in\X}$ denotes the quantum encoding of the classical random variable $X$.

\textit{Processing}: After encoding the data, we must process it using a quantum circuit. We model this by  an arbitrary quantum channel $\N:\S(\H)\rightarrow \S(\H')$. A quantum channel is a 
completely positive and trace preserving mapping~\cite{wilde2013quantum}.

\textit{Measurement}: After processing, we need to extract information from the quantum system $A$ using a positive operator-valued measure (POVM) $\F:=\{F_y\}_{y\in\Y}$, i.e., $F_y\in\P(\H')$ and $\sum_{y\in\Y} F_y=I$. Let random variable $Y\in\Y$ denote the outcome of the measurement. By Born's rule, we know that the probability of observing $Y=y\in\Y$ if the state of the quantum system $A$ is $\rho$ is given by $\mathbb{P}\{Y=y\}=\trace(F_y\rho)$. Therefore, $\mathbb{P}\{Y=y\,|\,X=x\}=\trace(F_y\rho^x)$ for all $x\in\X$ and $y\in\Y$.

\textit{Classical post-processing}: The last step is to process the random variable $Y$ by a classical processing routine to generate random variable $\widehat{Z}\in\Z$. This can be modelled by a conditional probability density function $\mathbb\{\widehat{Z}=z\,|\,Y=y\}=\gamma_{zy}$. Let $\gamma=(\gamma_{zy})_{z\in\Z,y\in\Y}$.

We denote the quantum computing procedure by tuple $(\Rho,\N,\F,\gamma)$. Before analysing the accuracy of the statistical inference model, we need to present maximal quantum leakage. We use this notion of information to bound the accuracy of inference models.

\begin{definition}[Maximal Quantum Leakage~\cite{Farokhi_PRA}] \label{def:qml} Maximal quantum leakage from random variable $X$ through quantum system $A$ is 
\begin{align}\label{eqn:def_qml}
    \mathcal{Q}(X\rightarrow A)_{\rho}
    =&\sup_{\{F_y\}_{y\in\Y}} \log\!\left(\sum_{y\in\Y} \max_{\substack{
         x\in\X: \\
         \mathbb{P}\{X=x\}>0
    }
    }\!\!\trace(\rho^x F_y) \!\right)\!,
\end{align}
where the supremum is taken over all POVMs $\F=\{F_y\}_{y\in\Y}$ with arbitrary finite outcome set $\Y$. 
\end{definition}

Maximal leakage was defined in~\cite{Farokhi_PRA} as a measure of private or secure information leakage to an arbitrary adversary for investigating security of quantum encoding policies (in communication, storage, and data analysis). Interestingly, maximal quantum leakage is independent of the output of the inference problem. It also does not depend of the distribution of the input; it only depends on the support set of the input and not the exact probability values.
In this letter, we show that maximal quantum leakage is  useful for bounding performance of quantum-based statistical inference mechanisms. This is established in the following theorem.

\begin{theorem} \label{tho:stat_inference} The accuracy of any quantum computing procedure $(\Rho,\N,\F,\gamma)$ is constrained as
    \begin{align*}
        \mathbb{P}\{\widehat{Z}=Z\}\leq \mathcal{Q}(X\rightarrow A)_{\rho}\max_{z\in\Z}\mathbb{P}\{Z=z\}.
    \end{align*}
\end{theorem}

\begin{proof}
Note that
\begin{align*}
    \frac{\mathbb{P}\{\widehat{Z}=Z\}}{\displaystyle \max_{z\in\Z}\mathbb{P}\{Z=z\}}
    &\leq 
    \sup_{\{F_y\}_{y\in\Y}}\sup_{Z,\widehat{Z}}\frac{\mathbb{P}\{\widehat{Z}=Z\}}{\displaystyle \max_{z\in\Z}\mathbb{P}\{Z=z\}}\\
    &=\mathcal{Q}(X\rightarrow A)_{\N(\rho)},
\end{align*}
where the equality follows from~\cite[Theorem~1]{Farokhi_PRA}. Subsequently, we have $\mathcal{Q}(X\rightarrow A)_{\N(\rho)}\leq \mathcal{Q}(X\rightarrow A)_{\rho}$~\cite[Proposition~3]{Farokhi_PRA}. This concludes the proof.
\end{proof}

The upper bound in Theorem~\ref{tho:stat_inference} shows that the accuracy of a quantum inference method is upper bounded by the maximal quantum leakage $\mathcal{Q}(X\rightarrow A)_{\rho}$, which is independent of the output and the distribution of the input of the statistical inference model, multiplied by $\max_{z\in\Z}\mathbb{P}\{Z=z\}$, which captures the accuracy of the best\footnote{Maximum \textit{a priori} estimator.} statistical inference algorithm that ignores the realization of the input $X$. The maximal quantum leakage $\mathcal{Q}(X\rightarrow A)_{\rho}$  essentially captures the multiplicative increase in the accuracy of the statistical inference model
by accessing the input $X$ via its quantum encoding $\mathcal{R}$. 

\begin{corollary} \label{cor:stats_inference} The accuracy of any quantum computing procedure $(\Rho,\N,\F,\gamma)$ is constrained as
    \begin{align*}
        \mathbb{P}\{\widehat{Z}=Z\}\leq &\max_{z\in\Z}\mathbb{P}\{Z=z\}\\
        &\times \min\{\log_2(|\X|),2\log_2(\dim(\H))\}.
    \end{align*}
\end{corollary}

\begin{proof}
    The proof follows from that $\mathcal{Q}(X\rightarrow A)_{\rho}\leq \min\{\log_2(|\X|),\log_2(\dim(\H)^2)\}$~\cite[Proposition~2]{Farokhi_PRA}.
\end{proof}

Note that, in Corollary~\ref{cor:stats_inference}, $\dim(\H)$ captures the dimension of the quantum system used for statistical inference. In fact, the quantum system will be composed of $\log_2(\dim(\H))$ qubits. If $2\log_2(\dim(\H))<\log_2(|\X|)$, the upper bound in Corollary~\ref{cor:stats_inference} is unnecessarily reduced by the dimension of the quantum system. This points to that the minimum number of qubits required for accurately solving an inference problem must be \textit{above} $\log_2(\dim(\H))\approx \log_2(|\X|)/2$. Note that we are not asserting that $\log_2(|\X|)/2$ is the optimal number of required qubits, but that this is a lower bound for how many qubits needed to solve an inference problem effectively. Furthermore, the only thing that matters here is the size of the support set of the input $X$ (not its distribution, not the output $Z$, not the quantum computing method used, and not the classical post-processing procedure implemented). Therefore, this results is rather universal and of value to any statistical inference problem.

The upper bound in Theorem~\ref{tho:stat_inference}, which is a function of the maximal quantum leakage $\mathcal{Q}(X\rightarrow A)_{\rho}$, only depends on the quantum encoding of the classical data $\Rho$. 
This hints that we can find a good universal encoding mechanism by maximizing $\mathcal{Q}(X\rightarrow A)_{\rho}$. This encoder can unlock the barrier in achieving a high accuracy in the statistical inference by increasing the upper bound in Theorem~\ref{tho:stat_inference}. This is pursued in what follows.

\textit{Optimal Quantum Encoding}:
 To compute the optimal universal encoder, we need to solve the following optimization problem:
\begin{align}\label{eqn:max_encoding}
    \argmax_{\substack{\rho^x\in\S(\H),\\ 
    \forall x\in\X}}\left[\sup_{\{F_y\}_{y\in\Y}} \log\left(\sum_{y\in\Y} \max_{\substack{
         x\in\X: \\
         \mathbb{P}\{X=x\}>0}
    } \trace(\rho^x F_y) \right)\right].
\end{align}
In the next proposition, we prove that this optimization problem attains its maximum over pure states, i.e., $\{\rho^x\}_{x\in\X}$ such that $\rank(\rho^x)=1$ for all $x\in\X$. This is an important revelation as often quantum computing procedures rely on pure states. Furthermore, this result is universal; it does not rely on the form of the statistical inference problem.

\begin{proposition} \label{prop:maximal}
    The maximum of~\eqref{eqn:max_encoding} is attained by pure states. 
\end{proposition}

\begin{proof}
    First, because $\log_2(\cdot)$ is strictly increasing, the optimization problem~\eqref{eqn:max_encoding} is equivalent to
    \begin{align*}
        \argmax_{\substack{\rho^x\in\S(\H),\\ 
    \forall x\in\X}}\; g(\{\rho^x\}_{x\in\X}),
    \end{align*}
    where
    \begin{align*}
        g(\{\rho^x\}_{x\in\X}):=\sup_{\{F_y\}_{y\in\Y}} \left(\sum_{y\in\Y} \max_{\substack{
         x\in\X: \\
         \mathbb{P}\{X=x\}>0}
    } \trace(\rho^x F_y)\right).
    \end{align*}
    It is easy to see that $g:\S(\H)^{|\X|}\rightarrow \mathbb{R}$ is convex because
    \begin{align*}
        g(\{\alpha &\rho^x+(1-\alpha)\sigma^x\}_{x\in\X})
        \\&= \sup_{\{F_y\}_{y\in\Y}}\sum_{y\in\Y} \max_{
         x\in\X} \trace((\alpha\rho^x+(1-\alpha)\sigma^x) F_y) 
         \\&\leq \alpha\sup_{\{F_y\}_{y\in\Y}}\sum_{y\in\Y} \max_{
         x\in\X} \trace(\rho^x F_y)\\
         &\quad+(1-\alpha)\sup_{\{F_y\}_{y\in\Y}}\sum_{y\in\Y} \max_{
         x\in\X} \trace(\sigma^x F_y) 
          \\&\leq \alpha g(\{\rho^x\}_{x\in\X})
         +(1-\alpha)g(\{\sigma^x\}_{x\in\X}).
    \end{align*}
    According to the Bauer's maximum principle~\cite[Theorem~3.5.29]{denkowski2003introduction}, originally proved in~\cite{bauer1958minimalstellen}, 
    $g(\cdot)$ must attain its maximum at an extreme point of $\S(\H)^{|\X|}$. The extreme points of this set are pure states~\cite[Theorem~2.3]{Holevo2013}. 
\end{proof}

Note that Proposition~\ref{prop:maximal} does not claim that the solution is unique. The problem might admit many solutions but at least one of those solutions uses  pure states for encoding the classical data. Note that all the optimal solutions have the same maximal quantum leakage.

\begin{proposition} \label{prop:index}
    If $\dim(\H)\geq |\X|$, the maximum of~\eqref{eqn:max_encoding} is attained by basis encoding, i.e., $\rho^x=\ket{\tau(x)} \bra{\tau(x)}$, where $\{\ket{i}\}_{i=1,\dots,\dim(\H)}$ is an orthonormal basis for $\H$ and $\tau:\X\rightarrow \{1,\dots,|\X|\}$ is any one-to-one mapping. 
\end{proposition}

\begin{proof}
    Note that $\mathcal{Q}(X\rightarrow A)\leq \log_2(|\X|)$ irrespective of $\{\rho^x\}_{x\in\X}$~\cite[Proposition~2]{Farokhi_PRA}. Let $\rho^x=\ket{\tau(x)} \bra{\tau(x)}$ for all $x\in\X$. Fix $\Y=\X$ and $F_y=\rho^y$ for all $y\in\Y$. We get $\sum_{y\in\Y} \max_{{
         x\in\X: 
         \mathbb{P}\{X=x\}>0}
    } \trace(\rho^x F_y)=|\X|$, which attains $\mathcal{Q}(X\rightarrow A)= \log_2(|\X|)$.
\end{proof}

Basis encoding, also called index encoding, is one of the most common forms of quantum encoding~\cite{10555535110653511068}. Proposition~\ref{prop:index} shows that this popular encoding policy is in fact universally optimal when there are enough qubits present. 

\begin{figure}
    \centering
    \includegraphics[width=0.95\columnwidth]{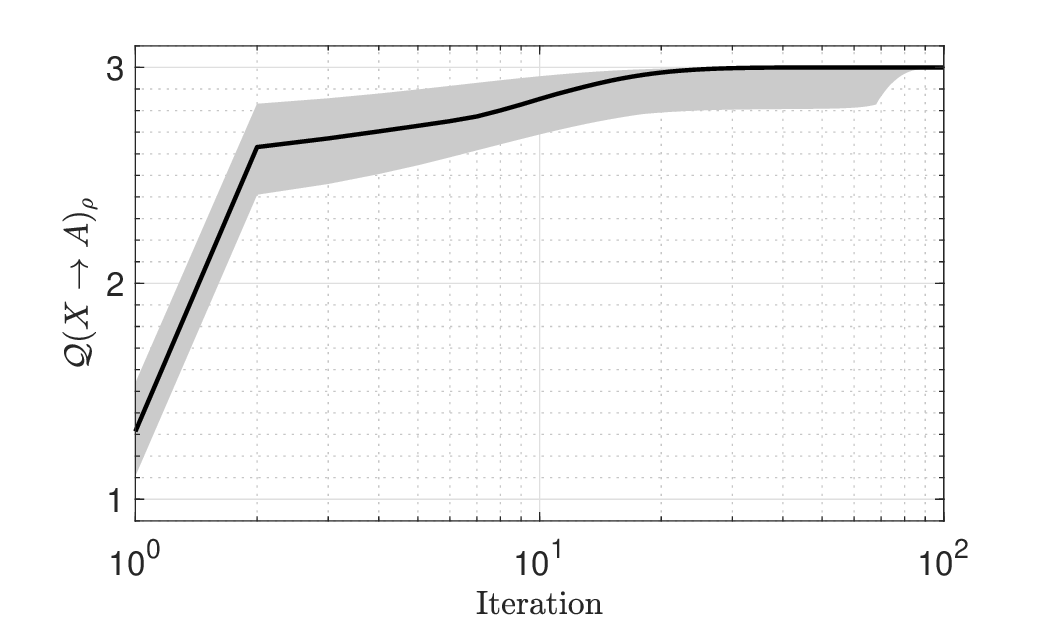}
    \caption{The maximal quantum leakage versus the number of iterations while implementing the projected subgradient ascent in~\eqref{eqn:gradient_ascent} to find the best quantum encoding of the classical data.}
    \label{fig:best_encoding}
\end{figure}

\begin{figure}
    \centering
    \includegraphics[width=0.95\columnwidth]{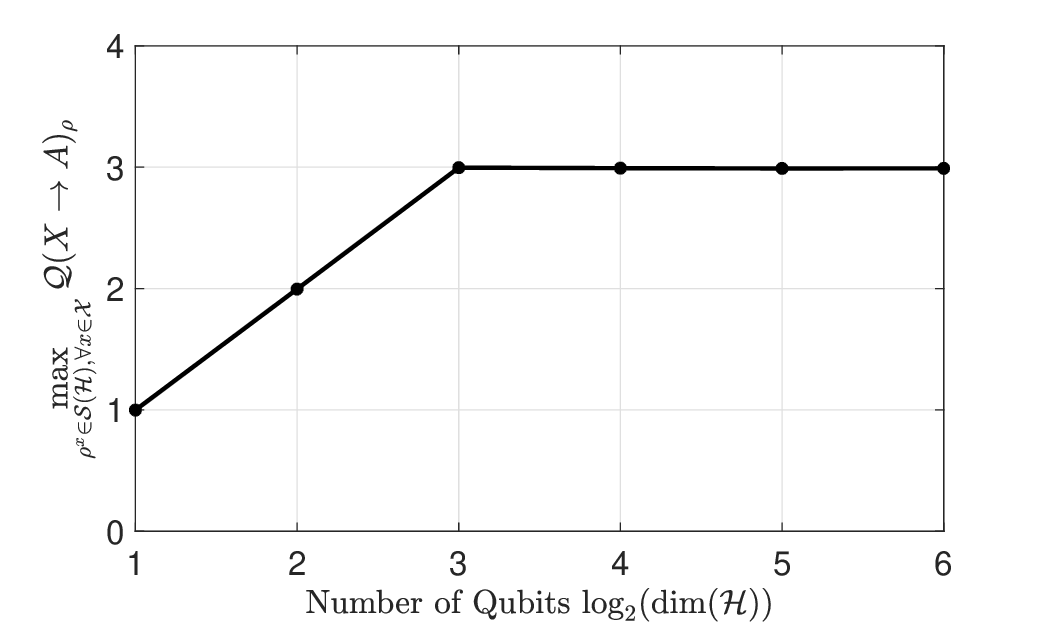}
    \caption{The Maximal quantum leakage 
    for the optimal universal encoder versus number of the qubits. 
    }
    \label{fig:qubits}
\end{figure}

\textit{Iterative Algorithm for Maximizing the Maximal Quantum Leakage}:
We make use of the iterative algorithm proposed in~\cite{Farokhi_PRA} to compute the maximal quantum leakage. Then, we update the quantum encoding via projected subgradient\footnote{Subgradients are generalizations of gradient to convex functions which are not necessarily differentiable.} ascent. We encourage interested readers to check~\cite[\chap\,14.2-14.3]{sun2006optimization} for more information on subgradients and non-smooth optimization. In~\cite{Farokhi_PRA}, it was shown that
\begin{subequations}\label{eqn:update_problem}
    \begin{align}
    2^{\mathcal{Q}(X\rightarrow A)_{\rho}}=\sup_{\{F_y\}} &\sum_{y\in\mathbb{Y}}  \trace(\rho^{x^*(y)} F_y),\\
   \mathrm{s.t.}\; & 0\preceq F_y,y\in\Y, \\
   &\sum_{y\in\Y} F_y=I,
\end{align}
\end{subequations}
where $\Y=\{1,\dots,\dim(\H)^2\}$ and 
\begin{align*}
    x^*(y)\in \argmax_{x\in\X} \trace(\rho_A^x F_y).
\end{align*}
Therefore, the subgradient with respect to $\rho^x$, $x\in\X$, is given by
\begin{align*}
    \partial_{\rho^x} 2^{\mathcal{Q}(X\rightarrow A)_{\rho}}=\sum_{
    \substack{
    y\in\Y:\\
    x^*(y)=x
    }} F_y^*,
\end{align*}
where $\partial_{\rho^x}$ denotes the subgradient with respect to $\rho^x$ and $\{F_y^*\}$ denotes the POVM that attains the maximum in~\eqref{eqn:update_problem}, which can be computed using the iterative algorithm in~\cite{Farokhi_PRA}. We can therefore update the quantum encoding of the data by using the projected subgradient ascent:
\begin{align}\label{eqn:gradient_ascent}
    \rho^x\gets \Pi[\rho^x+\mu\partial_{\rho^x} 2^{\mathcal{Q}(X\rightarrow A)_{\rho}}], 
\end{align}
where $\mu>0$ is the step size (selected small enough to ensure convergence) and $\Pi$ is projection to the set of rank-one matrices with unit trace. For any Hermitian operator $\sigma\in\L(\H)$, the projection $\Pi$ is defined as $\Pi[\sigma]=\ket{i^*}\bra{i^*},$ where $i^*\in\argmax_{1\leq i\leq \dim(\H)} |\lambda_i|$ with $\sigma=\sum_{i}\lambda_i \ket{i}\bra{i}$. Note that $\sigma$ can always be diagonalized as above because it is Hermitian. 

In what follows, we demonstrate the convergence of the subgradient ascent in the case where $\X:=\{1,\dots,8\}$ and $\dim(\H)=8$. In this case, Proposition~\ref{prop:index} shows that the basis encoding is the optimal universal encoding strategy. Therefore, for optimal universal encoding strategy, $\mathcal{Q}(X\rightarrow A)_{\rho}=3$. Figure~\ref{fig:best_encoding} illustrates the maximal quantum leakage versus the number of iterations while implementing the projected subgradient ascent in~\eqref{eqn:gradient_ascent} to find the best quantum encoding. The algorithm is initialized at a random encoding strategy. The gray area demonstrates the maximum and minimum in each iteration over 100 runs of the algorithm (note the randomness in the initialization) and the solid black line shows the median in each iteration. As we can see, the algorithm rapidly converges to the optimal quantum encoding. Note that the optimal encoder is not unique but $\mathcal{Q}(X\rightarrow A)_{\rho}=3$ irrespectively.

Figure~\ref{fig:qubits} shows the maximal quantum leakage for the optimal universal encoder versus the number of the qubits for the case where $\X=\{1,\dots,8\}$. Clearly, the maximal quantum leakage increases until the number of the qubits hits $\log_2(|\X|)=3$. Based on Corollary~\ref{cor:stats_inference}, we expected that the minimum number of required qubits must be above $\log_2(|\X|)/2=1.5$. This shows that there is a multiplicative gap of two between the lower bound for the number of qubits, discussed earlier, and the exact number of qubits needed. This can be attributed to the looseness of the bound in Corollary~\ref{cor:stats_inference} for pure states. Pursing a tighter bound is a viable direction for future research. Irrespective of this, the observation from Corollary~\ref{cor:stats_inference} seems to be optimal up to a constant scaling factor.

\textit{Discussions}: We presented a framework for developing optimal universal strategies for quantum encoding of classical data. The framework is universal as it is optimal for a wide array of statistical inference tasks; it does not rely on the output of the statistical inference problem and the underlying probability distributions. The optimal encoding strategy only takes into account the support set of the input variable. We proved that the optimal universal encoding strategy is attained by pure states. Furthermore, when there are enough qubits, basis encoding was proved to be universally optimal. For all other cases, an iterative method for numerically computing the optimal universal encoding strategy was presented. Future work can focus on using maximal quantum leakage for bounding the accuracy of inference models in the small-data regime and in quantum machine learning. 


\bibliography{apssamp}

\begin{thebibliography}{16}%
\makeatletter
\providecommand \@ifxundefined [1]{%
 \@ifx{#1\undefined}
}%
\providecommand \@ifnum [1]{%
 \ifnum #1\expandafter \@firstoftwo
 \else \expandafter \@secondoftwo
 \fi
}%
\providecommand \@ifx [1]{%
 \ifx #1\expandafter \@firstoftwo
 \else \expandafter \@secondoftwo
 \fi
}%
\providecommand \natexlab [1]{#1}%
\providecommand \enquote  [1]{``#1''}%
\providecommand \bibnamefont  [1]{#1}%
\providecommand \bibfnamefont [1]{#1}%
\providecommand \citenamefont [1]{#1}%
\providecommand \href@noop [0]{\@secondoftwo}%
\providecommand \href [0]{\begingroup \@sanitize@url \@href}%
\providecommand \@href[1]{\@@startlink{#1}\@@href}%
\providecommand \@@href[1]{\endgroup#1\@@endlink}%
\providecommand \@sanitize@url [0]{\catcode `\\12\catcode `\$12\catcode
  `\&12\catcode `\#12\catcode `\^12\catcode `\_12\catcode `\%12\relax}%
\providecommand \@@startlink[1]{}%
\providecommand \@@endlink[0]{}%
\providecommand \url  [0]{\begingroup\@sanitize@url \@url }%
\providecommand \@url [1]{\endgroup\@href {#1}{\urlprefix }}%
\providecommand \urlprefix  [0]{URL }%
\providecommand \Eprint [0]{\href }%
\providecommand \doibase [0]{https://doi.org/}%
\providecommand \selectlanguage [0]{\@gobble}%
\providecommand \bibinfo  [0]{\@secondoftwo}%
\providecommand \bibfield  [0]{\@secondoftwo}%
\providecommand \translation [1]{[#1]}%
\providecommand \BibitemOpen [0]{}%
\providecommand \bibitemStop [0]{}%
\providecommand \bibitemNoStop [0]{.\EOS\space}%
\providecommand \EOS [0]{\spacefactor3000\relax}%
\providecommand \BibitemShut  [1]{\csname bibitem#1\endcsname}%
\let\auto@bib@innerbib\@empty
\bibitem [{\citenamefont {Schuld}\ \emph {et~al.}(2021)\citenamefont {Schuld},
  \citenamefont {Sweke},\ and\ \citenamefont {Meyer}}]{PhysRevA103032430}%
  \BibitemOpen
  \bibfield  {author} {\bibinfo {author} {\bibfnamefont {M.}~\bibnamefont
  {Schuld}}, \bibinfo {author} {\bibfnamefont {R.}~\bibnamefont {Sweke}},\ and\
  \bibinfo {author} {\bibfnamefont {J.~J.}\ \bibnamefont {Meyer}},\ }\bibfield
  {title} {\bibinfo {title} {Effect of data encoding on the expressive power of
  variational quantum-machine-learning models},\ }\href@noop {} {\bibfield
  {journal} {\bibinfo  {journal} {Phys. Rev. A}\ }\textbf {\bibinfo {volume}
  {103}},\ \bibinfo {pages} {032430} (\bibinfo {year} {2021})}\BibitemShut
  {NoStop}%
\bibitem [{\citenamefont {Bennett}\ and\ \citenamefont
  {Brassard}(2014)}]{bennett2014quantum}%
  \BibitemOpen
  \bibfield  {author} {\bibinfo {author} {\bibfnamefont {C.~H.}\ \bibnamefont
  {Bennett}}\ and\ \bibinfo {author} {\bibfnamefont {G.}~\bibnamefont
  {Brassard}},\ }\bibfield  {title} {\bibinfo {title} {Quantum cryptography:
  Public key distribution and coin tossing},\ }\href@noop {} {\bibfield
  {journal} {\bibinfo  {journal} {Theoretical computer science}\ }\textbf
  {\bibinfo {volume} {560}},\ \bibinfo {pages} {7} (\bibinfo {year}
  {2014})}\BibitemShut {NoStop}%
\bibitem [{\citenamefont {Haselgrove}(2005)}]{haselgrove2005optimal}%
  \BibitemOpen
  \bibfield  {author} {\bibinfo {author} {\bibfnamefont {H.~L.}\ \bibnamefont
  {Haselgrove}},\ }\bibfield  {title} {\bibinfo {title} {Optimal state encoding
  for quantum walks and quantum communication over spin systems},\ }\href@noop
  {} {\bibfield  {journal} {\bibinfo  {journal} {Physical Review A}\ }\textbf
  {\bibinfo {volume} {72}},\ \bibinfo {pages} {062326} (\bibinfo {year}
  {2005})}\BibitemShut {NoStop}%
\bibitem [{\citenamefont {Weigold}\ \emph {et~al.}(2022)\citenamefont
  {Weigold}, \citenamefont {Barzen}, \citenamefont {Leymann},\ and\
  \citenamefont {Salm}}]{10555535110653511068}%
  \BibitemOpen
  \bibfield  {author} {\bibinfo {author} {\bibfnamefont {M.}~\bibnamefont
  {Weigold}}, \bibinfo {author} {\bibfnamefont {J.}~\bibnamefont {Barzen}},
  \bibinfo {author} {\bibfnamefont {F.}~\bibnamefont {Leymann}},\ and\ \bibinfo
  {author} {\bibfnamefont {M.}~\bibnamefont {Salm}},\ }\bibfield  {title}
  {\bibinfo {title} {Data encoding patterns for quantum computing},\ }in\
  \href@noop {} {\emph {\bibinfo {booktitle} {Proceedings of the 27th
  Conference on Pattern Languages of Programs}}},\ \bibinfo {series and number}
  {PLoP '20}\ (\bibinfo  {publisher} {The Hillside Group},\ \bibinfo {address}
  {USA},\ \bibinfo {year} {2022})\BibitemShut {NoStop}%
\bibitem [{\citenamefont {Weigold}\ \emph {et~al.}(2021)\citenamefont
  {Weigold}, \citenamefont {Barzen}, \citenamefont {Leymann},\ and\
  \citenamefont {Salm}}]{weigold2021encoding}%
  \BibitemOpen
  \bibfield  {author} {\bibinfo {author} {\bibfnamefont {M.}~\bibnamefont
  {Weigold}}, \bibinfo {author} {\bibfnamefont {J.}~\bibnamefont {Barzen}},
  \bibinfo {author} {\bibfnamefont {F.}~\bibnamefont {Leymann}},\ and\ \bibinfo
  {author} {\bibfnamefont {M.}~\bibnamefont {Salm}},\ }\bibfield  {title}
  {\bibinfo {title} {Encoding patterns for quantum algorithms},\ }\href@noop {}
  {\bibfield  {journal} {\bibinfo  {journal} {IET Quantum Communication}\
  }\textbf {\bibinfo {volume} {2}},\ \bibinfo {pages} {141} (\bibinfo {year}
  {2021})}\BibitemShut {NoStop}%
\bibitem [{\citenamefont {Elron}\ and\ \citenamefont
  {Eldar}(2007)}]{elron2007optimal}%
  \BibitemOpen
  \bibfield  {author} {\bibinfo {author} {\bibfnamefont {N.}~\bibnamefont
  {Elron}}\ and\ \bibinfo {author} {\bibfnamefont {Y.~C.}\ \bibnamefont
  {Eldar}},\ }\bibfield  {title} {\bibinfo {title} {Optimal encoding of
  classical information in a quantum medium},\ }\href@noop {} {\bibfield
  {journal} {\bibinfo  {journal} {IEEE transactions on information theory}\
  }\textbf {\bibinfo {volume} {53}},\ \bibinfo {pages} {1900} (\bibinfo {year}
  {2007})}\BibitemShut {NoStop}%
\bibitem [{\citenamefont {Korzekwa}\ \emph {et~al.}(2022)\citenamefont
  {Korzekwa}, \citenamefont {Pucha{\l}a}, \citenamefont {Tomamichel},\ and\
  \citenamefont {{\.Z}yczkowski}}]{korzekwa2022encoding}%
  \BibitemOpen
  \bibfield  {author} {\bibinfo {author} {\bibfnamefont {K.}~\bibnamefont
  {Korzekwa}}, \bibinfo {author} {\bibfnamefont {Z.}~\bibnamefont
  {Pucha{\l}a}}, \bibinfo {author} {\bibfnamefont {M.}~\bibnamefont
  {Tomamichel}},\ and\ \bibinfo {author} {\bibfnamefont {K.}~\bibnamefont
  {{\.Z}yczkowski}},\ }\bibfield  {title} {\bibinfo {title} {Encoding classical
  information into quantum resources},\ }\href@noop {} {\bibfield  {journal}
  {\bibinfo  {journal} {IEEE Transactions on Information Theory}\ }\textbf
  {\bibinfo {volume} {68}},\ \bibinfo {pages} {4518} (\bibinfo {year}
  {2022})}\BibitemShut {NoStop}%
\bibitem [{\citenamefont {Mitra}\ and\ \citenamefont
  {Mandayam}(2021)}]{mitra2021optimal}%
  \BibitemOpen
  \bibfield  {author} {\bibinfo {author} {\bibfnamefont {A.}~\bibnamefont
  {Mitra}}\ and\ \bibinfo {author} {\bibfnamefont {P.}~\bibnamefont
  {Mandayam}},\ }\bibfield  {title} {\bibinfo {title} {On optimal cloning and
  incompatibility},\ }\href@noop {} {\bibfield  {journal} {\bibinfo  {journal}
  {Journal of Physics A: Mathematical and Theoretical}\ }\textbf {\bibinfo
  {volume} {54}},\ \bibinfo {pages} {405303} (\bibinfo {year}
  {2021})}\BibitemShut {NoStop}%
\bibitem [{\citenamefont {Farokhi}\ and\ \citenamefont
  {Kim}(2024)}]{farokhi2024measuring}%
  \BibitemOpen
  \bibfield  {author} {\bibinfo {author} {\bibfnamefont {F.}~\bibnamefont
  {Farokhi}}\ and\ \bibinfo {author} {\bibfnamefont {S.}~\bibnamefont {Kim}},\
  }\bibfield  {title} {\bibinfo {title} {Measuring quantum information leakage
  under detection threat},\ }\href@noop {} {\bibfield  {journal} {\bibinfo
  {journal} {arXiv preprint arXiv:2403.11433}\ } (\bibinfo {year}
  {2024})}\BibitemShut {NoStop}%
\bibitem [{\citenamefont {Nielsen}\ and\ \citenamefont
  {Chuang}(2000)}]{nielsen2010quantum}%
  \BibitemOpen
  \bibfield  {author} {\bibinfo {author} {\bibfnamefont {M.}~\bibnamefont
  {Nielsen}}\ and\ \bibinfo {author} {\bibfnamefont {I.}~\bibnamefont
  {Chuang}},\ }\href@noop {} {\emph {\bibinfo {title} {Quantum Computation and
  Quantum Information}}}\ (\bibinfo  {publisher} {Cambridge University Press},\
  \bibinfo {year} {2000})\BibitemShut {NoStop}%
\bibitem [{\citenamefont {Wilde}(2013)}]{wilde2013quantum}%
  \BibitemOpen
  \bibfield  {author} {\bibinfo {author} {\bibfnamefont {M.}~\bibnamefont
  {Wilde}},\ }\href@noop {} {\emph {\bibinfo {title} {Quantum Information
  Theory}}},\ Quantum Information Theory\ (\bibinfo  {publisher} {Cambridge
  University Press},\ \bibinfo {year} {2013})\BibitemShut {NoStop}%
\bibitem [{\citenamefont {Farokhi}(2024)}]{Farokhi_PRA}%
  \BibitemOpen
  \bibfield  {author} {\bibinfo {author} {\bibfnamefont {F.}~\bibnamefont
  {Farokhi}},\ }\bibfield  {title} {\bibinfo {title} {Maximal information
  leakage from quantum encoding of classical data},\ }\href@noop {} {\bibfield
  {journal} {\bibinfo  {journal} {Phys. Rev. A}\ }\textbf {\bibinfo {volume}
  {109}},\ \bibinfo {pages} {022608} (\bibinfo {year} {2024})}\BibitemShut
  {NoStop}%
\bibitem [{\citenamefont {Denkowski}\ \emph {et~al.}(2003)\citenamefont
  {Denkowski}, \citenamefont {Mig{\'o}rski},\ and\ \citenamefont
  {Papageorgiou}}]{denkowski2003introduction}%
  \BibitemOpen
  \bibfield  {author} {\bibinfo {author} {\bibfnamefont {Z.}~\bibnamefont
  {Denkowski}}, \bibinfo {author} {\bibfnamefont {S.}~\bibnamefont
  {Mig{\'o}rski}},\ and\ \bibinfo {author} {\bibfnamefont {N.~S.}\ \bibnamefont
  {Papageorgiou}},\ }\href@noop {} {\emph {\bibinfo {title} {An Introduction to
  Nonlinear Analysis: Theory}}},\ An Introduction to Nonlinear Analysis\
  (\bibinfo  {publisher} {Kluwer Academic Publishers},\ \bibinfo {year}
  {2003})\BibitemShut {NoStop}%
\bibitem [{\citenamefont {Bauer}(1958)}]{bauer1958minimalstellen}%
  \BibitemOpen
  \bibfield  {author} {\bibinfo {author} {\bibfnamefont {H.}~\bibnamefont
  {Bauer}},\ }\bibfield  {title} {\bibinfo {title} {Minimalstellen von
  funktionen und extremalpunkte},\ }\href@noop {} {\bibfield  {journal}
  {\bibinfo  {journal} {Archiv der Mathematik}\ }\textbf {\bibinfo {volume}
  {9}},\ \bibinfo {pages} {389} (\bibinfo {year} {1958})}\BibitemShut {NoStop}%
\bibitem [{\citenamefont {Holevo}(2013)}]{Holevo2013}%
  \BibitemOpen
  \bibfield  {author} {\bibinfo {author} {\bibfnamefont {A.~S.}\ \bibnamefont
  {Holevo}},\ }\href@noop {} {\emph {\bibinfo {title} {Quantum Systems,
  Channels, Information: A Mathematical Introduction}}}\ (\bibinfo  {publisher}
  {De Gruyter},\ \bibinfo {address} {Berlin, Boston},\ \bibinfo {year}
  {2013})\BibitemShut {NoStop}%
\bibitem [{\citenamefont {Sun}\ and\ \citenamefont
  {Yuan}(2006)}]{sun2006optimization}%
  \BibitemOpen
  \bibfield  {author} {\bibinfo {author} {\bibfnamefont {W.}~\bibnamefont
  {Sun}}\ and\ \bibinfo {author} {\bibfnamefont {Y.~X.}\ \bibnamefont {Yuan}},\
  }\href@noop {} {\emph {\bibinfo {title} {Optimization Theory and Methods:
  Nonlinear Programming}}},\ Springer Optimization and Its Applications\
  (\bibinfo  {publisher} {Springer US},\ \bibinfo {year} {2006})\BibitemShut
  {NoStop}%
\end{thebibliography}%

\end{document}